\documentclass[journal]{IEEEtran}%
\usepackage{mathrsfs}
\usepackage[dvipsnames,usenames]{color}
\usepackage{amsmath,amsthm}
\usepackage{amssymb}
\usepackage{graphicx}
\usepackage{subfigure}
\usepackage{times}
\usepackage{graphics,color}
\usepackage{lipsum}
\newtheorem{theorem}{Theorem}
\newtheorem{lemma}{Lemma}

\newtheorem{remark}{Remark}
\newcommand{\onetom}{1,\cdots,m}

\newcommand{\R}{\mathbb R}

\title{Synchronization in Networks of Linearly Coupled Dynamical
Systems via Event-triggered Diffusions
\thanks{This work is jointly supported by the National Natural Sciences
Foundation of China under Grant Nos. 61273211 and 61273309, the
Marie Curie International Incoming Fellowship from the European Commission
(302421), and the Program for New Century Excellent Talents
in University (NCET-13-0139).} }
\author{Wenlian
Lu,~\IEEEmembership{Member,~IEEE}, Yujuan Han, Tianping
Chen,~~\IEEEmembership{Senior~Member,~IEEE}
\thanks{W. Lu is with the Centre for Computational Systems Biology and the School of Mathematical Sciences, Fudan University,
Shanghai 200433, China, and also with the Department of Computer
Science, the University of Warwick, Coventry CV4 7AL, United Kingdom. (e-mail: wenlian@fudan.edu.cn). }
\thanks{
Y. Han is with the College of Information Engineering, Shanghai Maritime
University, Shanghai, China. (e-mail: yjhan@shmtu.edu.cn).  }
\thanks{T. Chen is with  the School of Mathematical Sciences, Fudan University,
Shanghai, China, and also with the School of Computer Science, Fudan
University, Shanghai, China. (e-mail: tchen@fudan.edu.cn).
}}
\date{}

\begin{document}
\maketitle
\begin{abstract}
In this paper, we utilize event-triggered coupling configuration to realize
synchronization of linearly coupled dynamical systems. Here, the diffusion
couplings are set up from the latest observations of the nodes of its
neighborhood and the next observation time is triggered by the proposed
criteria based on the local neighborhood information as well. Two
scenarios are considered: continuous monitoring, that each node can
observe its neighborhood's instantaneous states, and discrete monitoring,
that each node can only obtain its neighborhood's states at the same
time point when the coupling term is triggered. In both cases,
we prove that if the system with persistent coupling can synchronize, then
these event-trigger coupling strategies can synchronize the system, too.
\end{abstract}
\section{Introduction}
\IEEEPARstart{S}{ynchronization} of coupled dynamical systems have been widely studied over
the past decades \cite{Wu}-\cite{Zhang2}, which can be
characterized by that all oscillators approach to a uniform dynamical
behavior and generally assured by the couplings among nodes and/or
external distributed and cooperative control.

In most existing works on linearly coupled dynamical systems, each node
needs to gather its own state and neighbors states and update them
spontaneously or in a fixed sampling rate, which may cost much. In order to
reduce the sampling rate of the coupling between nodes, specific
discretization is necessary. As pointed out by \cite{Astrom}, event-based
sampling was proved to possess better performance than sampling
periodically in time. Hence, \cite{Mazo}-\cite{Wang11} suggested that the
event-based control algorithms to reduce communication and computation load
in networked coupled systems. And \cite{Wang11}-\cite{HZ} showed that the
event-based control maintains the control performance. Event-based control
seems to be suitable for coupled dynamical systems with limited resources
and many works addressed the event-triggered algorithms. \cite{Dimarogonas}
considered centralized formulation, distributed formulation event-driven
strategies for multi-agent systems and proposed a self-triggered setup,
by which continuous measuring of the neighbor states can be avoided.
\cite{Johannesson}-\cite{Rabi} studied the stochastic event-driven
strategies. \cite{Seyboth} introduced event-based control strategies for both networks of
single-integrators with time-delay in communication links and networks of
double-integrators. By using scattering
transformation, \cite{YuHan} investigated the output synchronization
problem of multi-agent systems with event-driven communication in the
presence of constant communication delays. In some cases, event-driven strategies
for multi-agent systems can be regarded as linearization and discretization
process. For example, as mentioned in \cite{LC2004,LC2007}, the following
algorithm
\begin{align}
x^{i}(t+1)=f(x^{i}(t))+c_{i}\sum_{j=1}^{m}a_{ij}(f(x^{j}(t)))\label{cml}
\end{align}
can be a variant of the event triggering (distributed, self triggered) model for
consensus problem. In centralized control, the bound for
$(t_{k+1}^{i}-t_{k}^{i})=(t_{k+1}-t_{k})$ to reach synchronization was
given in the paper \cite{LC2004} when the coupling graph is indirected and
in \cite{LC2007} for the directed coupling graph.

Motivated by these works, we apply the idea of event-triggered sampling
scheme to the coupling configurations to realize synchronization of linearly coupled
dynamical systems. Here, for each node, the coupling term is set
up from the information of its local neighborhood at the last event time
and the event is triggered by some criteria derived from the information of
its local neighborhood. That is, once the triggering rule of node is
satisfied, the coupling term of this node is updated. Thus, the coupling terms are piece-wise constant
between two neighboring event times. We consider two scenarios:
continuous monitoring and discrete monitoring. Continuous monitoring means
that each node can observe its neighborhood's instantaneous information,
but discrete monitoring means that each node can only obtain the its
neighborhood's information at this event triggered time. As a payoff for
small cost of discrete monitoring, the triggering events happen more
frequently than continuous-time monitoring. For each scenario, it is shown
that the proposed event-triggered strategies guarantee the performances of
the nominal systems.

This paper is organized as follows. In sec. II, we propose event-trigger
coupling strategies to guarantee synchronization by employing continuous
monitoring. In sec. III, we consider discrete monitoring. Simulations are
given in sec. IV to verify the theoretical results. We conclude this paper
in sec. V.

\section{Continuous-time Monitoring}\label{2.1}

We consider the following network of coupled dynamical systems with piece-wise constant linear couplings:
\begin{eqnarray}
\frac{dx^{i}(t)}{dt}=f(x^{i}(t))-c\sum_{j=1}^{m}L_{ij}\Gamma(x^{j}(t_{k}^{i})-x^{i}(t_{k}^{i})),\nonumber\\
~~t_{k}^{i}\le t<t_{k+1}^{i},~i=\onetom.\label{cds}
\end{eqnarray}
Here, $x^{i}(t)\in\R^{n}$ denotes the state vector of node $i$, the
continuous map $f(\cdot):\R^{m}\to\R^{m}$ denotes the identical node
dynamics if there is no coupling. $L=[L_{ij}]_{i,j=1}^{m}\in\R^{m,m}$ is
the Laplacian matrix of the underlying bi-graph $\mathcal G=\{V,E\}$, with
the node set $V$ and link set $E$: for each pair of nodes $i\ne j$,
$L_{ij}=-1$ if $i$ is linked to $j$ otherwise $L_{ij}=0$, and
$L_{ii}=-\sum_{j=1}^{m}L_{ij}$; the graph that we consider in this paper is undirected and connected,
so $L$ is irreducible and symmetric. $c$ is the uniform coupling strength
at all nodes, and $\Gamma\in\R^{n,n}$ denotes the inner configuration matrix.
Let $0=\lambda_{1}(L)<\lambda_{2}(L)\le\cdots\le\lambda_{m}(L)$ be the
eigenvalues of $L$ with counting the multiplicities.

The increasing triggering event time sequence
$\{t_{k}^{i}\}_{k=1}^{\infty}$ (to be defined) are node-wise for
$i=\onetom$. At time $t$, each node $i$ collects its neighbor's state with
respect to an identical time point $t_{k_{i}(t)}^{i}$ with
$k_{i}(t)=arg\max_{k'}\{t^{i}_{k'}\le t\}$.

For the node dynamics map $f$, we suppose it belong to some map class
$Quad(P,\alpha\Gamma,\beta)$ for some positive definite matrix
$P\in\R^{n,n}$, constant $\alpha\in\R$ and positive constant $\beta>0$,
i.e.,
\begin{align}
\nonumber&(u-v)^{\top}P\bigg[f(u)-f(v)-\alpha\Gamma(u-v)\bigg]\\
&\le-\beta(u-v)^{\top}P(u-v)\label{quad}
\end{align}
holds for all $u,v\in\R^{n}$. In fact, we do not need this $Quad$ condition
(\ref{quad}) holds for all $u,v\in\R^{n}$ but for a region
$\Omega\subset\R^{n}$, which is assumed to contain a global attractors of
the coupling systems (\ref{cds}).

We highlight the basic idea behind the setup of the coupling term above. Instead of using the spontaneous state from the neighborhood to
realize synchronize, an economic alternative for the node $i$ is to use the
neighbor's constant states at the nearest time point $t^{i}_{k}$ until some
pre-defined event is triggered at time $t^{i}_{k+1}$; then the incoming
neighbor's information is updated by the states at $t^{i}_{k+1}$ until the
next event is triggered, and so on. The event is defined based on the
neighbor's and its own states with some prescribed rule. This process goes
on through all nodes in a parallel fashion.

To depict the event that triggers the next coupling time point, we
introduce the following candidate Lyapunov function:
\begin{align}
\nonumber V(t)& =-\frac{1}{2}\sum_{i=1}^{m}
\sum_{j=1}^{m}L_{ij}(x^{i}(t)-x^{j}(t))^{\top}P(x^{i}(t)-x^{j}(t))\\
&=x^{\top}(t)(L\otimes P)x(t)
\label{V}
\end{align}
with $x=[{x^{1}}^{\top}(t),\cdots,{x^{m}}^{\top}(t)]^{\top}\in\R^{nm}$ and
$\otimes$ represents the Kronecker product. For a compact expression, we
denote $F(x)=[f(x^{1})^{\top},\cdots,f(x^{m})^{\top}]^{\top}$. Then, the
derivative of $V(t)$ along (\ref{cds}) is
\begin{align}
\nonumber\frac{d}{dt}V(t)\left|_{(\ref{cds})}\right.=&2x^{\top}(t)(L\otimes P)\bigg\{F(x(t))-\alpha (I_m\otimes\Gamma)x(t)\\
&-[(c L-\alpha I_m)\otimes\Gamma] x(t)+ce(t)\bigg\}
\label{dV1}
\end{align}
where $e(t)=[e_1^{\top}(t),\cdots,e^{\top}_m(t)]^{\top}$ and
$$e_i(t)=\sum_{j}L_{ij}\Gamma\left(x^j(t)-x^j(t_{k_i(t)}^i)-x^i(t)+x^i(t_{k_i(t)}^i)\right).$$
By assuming $f\in Quad(P,\alpha\Gamma,\beta)$, we have
\begin{align}\label{a}
&\frac{d}{dt}V(t)\left|_{(\ref{cds})}\right.\nonumber\\
\le&-2\beta x^{\top}(L\otimes P)x-2x^{\top}[L(c L-\alpha I_{m})\otimes(P\Gamma)]x\nonumber\\
&+2cx^{\top}(L\otimes P)e\nonumber\\
\le&-2\beta' x^{\top}(L\otimes P)x-2(\beta-\beta')x^{\top}(L\otimes P)x\nonumber\\
&-2(c\lambda_{2}(L)-\alpha)x^{\top}[L\otimes(P\Gamma)]x+2cx^{\top}(L\otimes P)e
\end{align}
with any $0<\beta'<\beta$. Pick a constant $\upsilon>0$, then
\begin{align}\label{b}
\nonumber2x^{\top}(L\otimes P)e&\leq {\upsilon}x^{\top}(L^2\otimes P^2)x+\frac{1}{\upsilon} e^{\top}e\\
&\leq {\upsilon\lambda_m(L)\lambda_m(P)}x^{\top}(L\otimes P)x+\frac{1}{\upsilon}e^{\top}e.
\end{align}
Substitute inequality (\ref{b}) into (\ref{a}), we have
\begin{align}\label{c}
\nonumber \frac{d}{dt}V(t)\left|_{(\ref{cds})}\right.&\le -2(\beta-\beta')V(t)
-2(c\lambda_{2}(L)-\alpha)V(t)\\
&+[-2\beta'+
{c\upsilon\lambda_m(L)\lambda_m(P)}]V(t)+\frac{c}{\upsilon}e^{\top}e.
\end{align}
Denote $$z_i(t)=\sqrt{\sum_{j\neq
i}(-L_{ij})(x^j(t)-x^i(t))^{\top}P(x^j(t)-x^i(t))}.$$ Then, we have
$V(t)=\sum_{i=1}^m z_i^2(t)$.

Denote $\|\cdot\|$ the Euclidean norm, i.e., for any vector $\xi\in\mathbb
R^n$, $\|\xi\|=\sqrt{\xi_1^2+\cdots+\xi_n^2}$. For a matrix $A\in\mathbb
R^{n\times n}$, the spectral norm of $A$ is induced from Euclidean norm,
i.e., $\|A\|=\sqrt{\lambda_{max}(A^{\top}A)}$. Hence, $\|A\xi\|\leq
\|A\|\|\xi\|$ always holds, which will be used later as default. Moreover,
$\|x\|_{P}=\sqrt{x^{\top}Px}$ for some positive definite matrix
$P\in\R^{n,n}$. Thus, we have the following theorem
\begin{theorem}\label{thm1}
Suppose that $f\in Quad(P,\alpha\Gamma,\beta)$ with
positive matrix $P$ and $\beta>0$, $c\lambda_{2}(L)>\alpha$ and
$P\Gamma$ is semi-positive definite. Pick $\beta'<\beta$. Then either one of
the following two updating rules can guarantee that system (\ref{cds})
synchronizes:
\begin{enumerate}
\item Set $t^{i}_{k+1}$ as the time point by the rule
\begin{align}
t_{k+1}^{i}=&\max_{\tau}\left\{\tau\ge t_{k}^{i}:\right.\nonumber\\
&\left.~\|e_i(\tau)\|
\le\frac{\beta'}{c\sqrt{\lambda_m(L)\lambda_m(P)}}z_i(\tau)\right\};
\label{event1}
\end{align}

\item Set $t^{i}_{k+1}$ as the time point by the rule
\begin{eqnarray}
t_{k+1}^{i}=\max\left\{\tau\ge t_{k}^{i}:~\|e_i(\tau)\|\le a\exp{(-b\tau)}\right\}.
\label{event2}
\end{eqnarray}
\end{enumerate}
\end{theorem}
\begin{proof}
Noting that $c\lambda_{2}(L)>\alpha$ holds and $P\Gamma$ is semi-positive definite. By (\ref{c}), we have
\begin{align}\label{Vcase1}
\nonumber\frac{d}{dt}V(t)\left|_{(\ref{cds})}\right.\leq& -2(\beta-\beta')V(t)+[c\upsilon\lambda_m(L)\lambda_m(P)\\
\nonumber&-2\beta']V(t)+\frac{c}{\upsilon}\|e(t)\|^2\\
\nonumber=& -2(\beta-\beta')V(t)+\sum_{i=1}^m [c\upsilon\lambda_m(L)\lambda_m(P)\\
&-2\beta']z_i^2(t)+\sum_{i=1}^m\frac{c}{\upsilon}\|e_i(t)\|^2
\end{align}
holds for any $\upsilon>0$.

(1).
In case of
\begin{equation}
\label{est1}
\|e_i(t)\|^2\leq \frac{\upsilon}{c}[2\beta'-c\upsilon\lambda_m(L)\lambda_m(P)]z_i^2(t)
\end{equation}
for some constant $\upsilon>0$, we have
\begin{equation}
\label{dV}
\frac{d}{dt}V(t)\left|_{(\ref{cds})}\right.\leq -2(\beta-\beta')V.
\end{equation}
This implies that $V(t)$ converges to $0$ exponentially. Note
$$\max_{\upsilon>0}\frac{\upsilon}{c}[2\beta'-c\upsilon\lambda_m(L)\lambda_m(P)]
=\frac{\beta'^{2}}{c^{2}\lambda_{m}(L)\lambda_m(P)}.$$ Then, we take
$\upsilon=\frac{\beta'}{c\lambda_{m}(L)\lambda_m(P)}$, which guarantees that (\ref{event1})
holds.

(2). In case of
\begin{equation}
\label{est2}
\|e_i(t)\|\leq a\exp{(-bt)},
\end{equation}
we have
\begin{align*}
\frac{d}{dt}V(t)\left|_{(\ref{cds})}\right.\leq& -2(\beta-\beta')V+[c\upsilon\lambda_m(L)\lambda_m(P)\\
&-2\beta']V
+\frac{a^2cm}{\upsilon} \exp{(-2bt)}.
\end{align*}
Pick $\upsilon= \frac{2\beta'}{c\lambda_m(L)\lambda_m(P)}$. Then, we have
\begin{align}\label{Vcase2}
\nonumber\frac{d}{dt}V(t)\left|_{(\ref{cds})}\right.&\leq -2(\beta-\beta')V\\
&+\frac{a^2c^2m\lambda_m(L)\lambda_m(P)}{2\upsilon\beta'} \exp{(-2bt)},
\end{align}
which implies that $V(t)$ converges to $0$ exponentially.

\end{proof}
In fact, in (\ref{event1}), if $\tau=t_k^i$ but the system does not synchronize, then the left-hand term
$\|e_i(t_k^i)\|=0$ and at least one agent has positive right-hand term,
which means the next inter-event interval of one agent must be positive.
While in (\ref{event2}), if $\tau=t_k^i$ but the system does not synchronize, the left-hand term equals $0$ and
the right-hand term is positive for all nodes. Hence, the inter-event intervals of all
agents are positive. But in case the derivative of $e_i(t)$ is sufficiently
large, the inter-event interval might tends to $0$. Since the
dynamics of $e_i(t)$ is highly related to the property of $f(\cdot)$, towards a lower-bound of the inter-event intervals, we suppose that $f(\cdot)$ is
Lipschitz in the following theorem.

It should be highlighted that if $f$ is Lipschitz and $P\Gamma$ is semi-positive
definite, then $f\in Quad(P,\alpha\Gamma,\beta)$ with
$\beta=\frac{\alpha\lambda_1(P\Gamma)}{\lambda_m(P)}-\frac{\lambda_m(P)}{\lambda_1(P)}$.
Thus, we have

\begin{theorem}\label{thm2}
Suppose $f\in Quad(P,\alpha\Gamma,\beta)$ with positive matrix $P$ and
$\beta>0$, satisfies Lipschitz condition with Lipschitz constant $L_f$,
and there exists some $\sigma$ (possibly negative) such that
\begin{eqnarray}
(u-v)^{\top}P(f(u)-f(v))\ge\sigma (u-v)^{\top}P(u-v)\label{vi}
\end{eqnarray}
for all $u,v\in\mathbb R^{n}$.  $c\lambda_{2}(L)>\alpha$ and $P\Gamma$ is semi-positive
definite. For any $\beta'<\beta$, any initial condition and any time $t\leq
0$, we have
\begin{enumerate}
\item
With the updating rule (\ref{event1}), at least one agent has next
inter-event interval, which is lower-bounded by a common constant
$\tau_{O}>0$. 
in addition, if there exists $\varsigma>0$ such that $z_{i}^{2}(t)\ge
\varsigma V(t)$ for all $i=\onetom$ and $t\ge 0$, then the next inter-event
interval of every agent is strictly positive and is lower-bounded by a
common constant.
\item With the updating rule (\ref{event2}), the next inter-event interval of every
agent is strictly positive and is lower-bounded by a common constant.
\end{enumerate}
\end{theorem}
\begin{proof}
(1).  Note
\begin{align*}
\dot{x}^j(t)-\dot{x}^i(t)=&f(x^j)-f(x^i)-c[(L\otimes\Gamma)x(t_{k_j(t)}^j)]_j\\
&+c[(L\otimes\Gamma)x(t_{k}^i)]_i,~~~~~t\in[t_k^i,t_{k+1}^i).
\end{align*}
Combining with the facts that $f$ is Lipschitz and $V(t)$ is decreasing, we have
\begin{equation}\label{deri_dis}
\|\dot{x}^j(t)-\dot{x}^i(t)\|\leq L_f\|x^j(t)-x^i(t)\|+c\|\Gamma\|V(t_k^i).
\end{equation}
According to
\begin{equation}\label{dyn_e}
e_i(t)=\sum_{j}L_{ij}\Gamma\int_{t_k^i}^t[\dot{x}^j(s)-\dot{x}^i(s)]ds
\end{equation}
and inequality (\ref{deri_dis}), we have
\begin{align}\label{e}
\|e_i(t)\|\leq
&\frac{1}{\sqrt{\lambda_1(P)}}\left(L_f\|\Gamma\|+mc\|\Gamma\|^2\right)\sqrt{2V(t_k^i)}(t-t_k^i).
\end{align}
And, noting $V(t)=\sum_{i=1}^{m}z_{i}^{2}(t)$, there exists $i_{*}$ such that
\begin{eqnarray}\label{z}
z_{i_{*}}(t)\ge\sqrt{\frac{1}{m}V(t)}.
\end{eqnarray}
From the condition (\ref{vi}), (\ref{dV1}) gives
\begin{eqnarray}
\dot{V}&\ge&2\sigma x^{\top}(t)(L\otimes P)x(t)-2cx^{\top}(t)[L^{2}\otimes(P\Gamma)^{s}]x(t)
\nonumber\\
&&+c2x^{\top}(t)(L\otimes P)e(t).\label{xx1}
\end{eqnarray}
Noting
\begin{eqnarray*}
u^{\top}[L^{2}\otimes(P\Gamma)^{s}]u\le \lambda_{m}(L)\frac{\|(P\Gamma)^{s}\|}{\lambda_{1}P)}u^{\top}[L\otimes(P)]u,
\end{eqnarray*}
for all $u\in\R^{mn}$,
\begin{eqnarray*}
2x^{\top}(L\otimes P)e&\geq& -{\upsilon}x^{\top}(L^2\otimes P^2)x-\frac{1}{\upsilon} e^{\top}e\\
&\geq& -{\upsilon\lambda_m(L)\lambda_m(P)}x^{\top}(L\otimes P)x-\frac{1}{\upsilon}e^{\top}e,
\end{eqnarray*}
and event (\ref{event1}), (\ref{xx1}) gives
\begin{eqnarray*}
\dot{V}(t)\ge \varpi V(t)
\end{eqnarray*}
for all $t$ before the next triggering time, with
\begin{eqnarray*}
\varpi&=&2\sigma-2\lambda_{m}(L)\frac{\|(P\Gamma)^{s}\|}{\lambda_{1}(P)}-\upsilon\lambda_{m}(L)
\lambda_{m}(P)\\
&&-\frac{1}{\upsilon}\frac{\beta'}{c\lambda_{m}(L)\lambda_{m}(P)}.
\end{eqnarray*}
Thus, we have $V(t)\ge V(t_{k}^{i_{*}})\exp(\varpi (t-t^{i_{*}}_{k}))$. Combined with (\ref{e}) and (\ref{z}), this implies that for each $t\le t^{i_{*}}_{k}+\tau_{O}$, where $\tau_{O}$ satisfies
\begin{eqnarray*}
\frac{\left(L_f\|\Gamma\|+mc\|\Gamma\|^2\right)\sqrt{2}}{\sqrt{\lambda_1(P)}}\tau_{O}
=\frac{\beta'\sqrt{\frac{1}{m}\exp(\varpi\tau_{O})}}{c\sqrt{\lambda_m(L)\lambda_m(P)}}
\end{eqnarray*}
we have the inequality in (\ref{event1}) holds. Therefore, we have the next triggering time should be larger than $t^{i_{*}}_{k}+\tau_{O}$.

In addition, if $z_{i}^{2}(t)\ge \varsigma V(t)$ for all $i=\onetom$ and $t\ge 0$, replace
(\ref{z}) by
\begin{eqnarray}
z_{i}(t)\ge \sqrt{\varsigma V(t)}
\end{eqnarray}
for all $i=\onetom$. Then following the same arguments after (\ref{z}), we can conclude that we have the inequality in (\ref{event1}) holds for all $t\ge t_{k}^{i}+\tau_{O'}$ with some positive $\tau_{O'}$ satisfying:
\begin{eqnarray*}
\frac{\left(L_f\|\Gamma\|+mc\|\Gamma\|^2\right)\sqrt{2}}{\sqrt{\lambda_1(P)}}\tau_{O'}
=\frac{\beta'\varsigma\sqrt{\exp(\varpi\tau_{O'})}}{c\sqrt{\lambda_m(L)\lambda_m(P)}}.
\end{eqnarray*}



(2). Under updating rule (\ref{event2}).
By inequality (\ref{Vcase2}), we get
$$
V(t)\leq \rho \exp{[-2\min(b,\beta-\beta')t]}
$$
with $\rho=V(0)+\frac{a^2c^2m\lambda_m(L)\lambda_m(P)}{4v\beta'(\beta-\beta')}+1$.  Hence, combined with (\ref{e}), this gives
\begin{align*}
\|e_i(t)\|\leq &\sqrt{\frac{2\rho}{\lambda_1(P)}} \left(L_f\|\Gamma\|+mc\|\Gamma\|^2\right)\\
&\times \exp{[-\min(b,\beta-\beta')t_k^i]}(t-t_k^i).
\end{align*}
Therefore, (\ref{event2}) is guaranteed by the following inequality
\begin{align*}
&\sqrt{\frac{2\rho}{\lambda_1(P)}} \left(L_f\|\Gamma\|+mc\|\Gamma\|^2\right)
\exp{[-\min(b,\beta-\beta')t_k^i]}(t-t_k^i)\\
& \leq a\exp{(-bt)}.
\end{align*}
Since at time $t=t_k^i$, $e_i(t)=0$ holds. Based on rule (\ref{event2}),
the next event will not trigger until $e_i(t)=a\exp{(-bt)}$. Thus, the
inter-event intervals $\tau=t_{k+1}^i-t_k^i$ is lower bounded by the
solution $\tau_D$ of the following equation
\begin{align*}
&\sqrt{\frac{2\rho}{\lambda_1(P)}} \left(L_f\|\Gamma\|+mc\|\Gamma\|^2\right)
\exp{[-\min(b,\beta-\beta')t_k^i]}\tau_D\\
&= a\exp{[-b(\tau_D+t_k^i)]}.
\end{align*}
It can be seen that this equation has a positive solution. This completes the proof.
\end{proof}
\begin{remark}
The updating rules (\ref{event1}) and (\ref{event2}) are different but
closely related to each other in some respects. It can be seen from
inequalities (\ref{Vcase1}) and (\ref{Vcase2}) used in the derivation that
the convergence behavior for (\ref{event1}) might be better than
(\ref{event2}). However, it makes rule (\ref{event1}) more complicated than
(\ref{event2}), since each agent should receive the message of the states
of its neighborhood but rule (\ref{event2}) does not need. Therefore, rule
(\ref{event1}) costs more updating times than (\ref{event2}). Moreover, as
shown by Theorem \ref{thm2}, rule (\ref{event2}) can guarantee the
positivity of the intervals to next updating time for all agents but rule
(\ref{event1}) can only guarantee it for at least one agent at each time or
for all nodes under some specific additional conditions.
\end{remark}
\section{Discrete-time Monitoring}\label{2.2}
By the discrete-time monitoring strategy, each node $i$ only needs its
local neighborhood's state at time-points $t^{i}_{k}$, $k=1,2,\cdots$. By
this way, the design of the next $t^{i}_{k+1}$ depends only on the local
states at time $t^{i}_{k}$, other than the triggering event (\ref{event1}),
(\ref{event2}), which requires the continuous time states. For early works, see \cite{LC2004,LC2007} for reference. 

Consider system (\ref{cds}) and  the candidate Lyapunov
function $V(x)$ with its derivative (\ref{dV1}). To propose a triggering criterion,
which depends only on $t_{k}^{i}$ by the criterion (\ref{event1}) in
Theorem \ref{thm1}, we need to estimate the bounds of
$(x^{q}(t)-x^{q}(t_{k}^{l}))-(x^{i}(t)-x^{i}(t_{k}^{l}))$ for any $q,i$
with $L_{iq}\neq 0$ and $x^{i}(t)-x^{j}(t)$ for any $i\ne j$.

First, we estimate the lower-bound of $x_i(t)-x_j(t)$, which satisfies
\begin{eqnarray*}
\frac{d[x^{i}(t)-x^{j}(t)]}{dt}=[f(x^{i}(t))-f(x^{j}(t))]+\theta_{i}(t_{k_i(t)}^i)-\theta_{j}(t_{k_j(t)}^j)
\end{eqnarray*}
by provided the initial values at $t^{i}_{k}$: $x^{i}(t^{i}_{k})$ and $x^{j}(t^{i}_{k})$.
This can be generalized as
\begin{eqnarray}
\begin{cases}\frac{du}{dt}=f(u(t))+\theta&u(0)=u_{0}\\
\frac{dv}{dt}=f(v(t))+\vartheta&v(0)=v_{0}.
\end{cases}\label{r2}
\end{eqnarray}

Suppose that the solutions  satisfy the following inequality:
\begin{eqnarray}
\bigg[(u(t)-v(t))^{\top}{P}(u(t)-v(t))\bigg]^{1/{2}}\geq\varrho(t,\theta,\vartheta,u_{0},v_{0})
\end{eqnarray}
Here $\varrho$ can be regarded as the lower-bound estimation of the distance (in $P$-norm) between two trajectories:
\begin{align*}
&\|u(t)-v(t)\|_{P}\\
=&\bigg\|\int_{0}^{t}[f(u(s))-f(v(s))]ds
+(\theta-\vartheta)t
-(u_0-v_0)\bigg\|_{P}.
\end{align*}

To specify $\varrho$, the celebrated Gronwall-Bellman inequality is used,
which can be verified straightforwardly and described as follows:
\begin{lemma} \cite{Gron,Bell}
For a nonnegative differentiable function $x(t)$ and two continuous functions: $\alpha(t)$ and $\beta(t)$, defined on interval $[a,b]$, if
\begin{eqnarray*}
\dot{x}\ge\alpha(t)x(t)+\beta(t)
\end{eqnarray*}
for all $t\in[a,b]$, then we have
\begin{eqnarray}
&&x(t)\ge x(a)\exp\left(\int_{a}^{t}\alpha(s)ds\right)\nonumber\\
&&+\int_{a}^{t}\beta(s)
\exp\left(\int_{s}^{t}\alpha(u)du\right);\label{gw1}
\end{eqnarray}
for two continuous functions $x(t)$ and $\beta(t)$, and an integrable
function $\alpha(t)$, defined on interval $[a,b]$, if $\beta(\cdot)$ is
nonnegative and
\begin{eqnarray*}
x(t)\le \alpha(t)+\int_{a}^{t}\beta(s)x(s)ds
\end{eqnarray*}
for all $t\in[a,b]$, then we have
\begin{eqnarray}
x(t)\le \alpha(t)+\int_{a}^{t}\alpha(s)\beta(s)\exp\left(\int_{s}^{t}\beta(u)du\right).\label{gw2}
\end{eqnarray}
\end{lemma}
By this lemma, $\varrho$ is a nonnegative-valued continuous map and satisfies (i). $\varrho(\cdot,\theta,\theta,u_{0},u_{0})\equiv 0$; (ii). $\varrho(0,\cdot,\cdot,u_{0},u_{0})\equiv 0$.
For example, assuming that condition (\ref{vi}) holds,
we have
\begin{align*}
&\frac{d}{dt}[(u(t)-v(t))^{\top}P(u(t)-v(t))]\left|_{(\ref{r2})}\right.\\
=&2(u-v)^{\top}P[f(u)-f(v)+\theta-\vartheta]\\
\ge&{2\sigma}(u-v)^{\top}P(u-v)-\mu (u-v)^{\top}P(u-v)\\
&-\frac{1}{\mu}(\theta-\vartheta)^{\top}P(\theta-\vartheta)
\end{align*}
for any $\mu>0$. By the Gronwall-Bell inequality (\ref{gw1}), we have
\begin{align*}
&(u(t)-v(t))^{\top}P(u(t)-v(t))\\
&\ge\exp{[({2\sigma}-\mu)t]}(u_{0}-v_{0})^{\top}P(u_{0}-v_{0})\\
&-\frac{(\theta-\vartheta)^{\top}P(\theta-\vartheta)/\mu}{2\sigma-\mu}\bigg\{\exp[({2\sigma}-\mu)t]-1\bigg\}
\end{align*}
which is positive for a small interval of $t$, starting from $0$, and $u_{0}\ne v_{0}$.

It can be seen that (\ref{vi}) holds for a large class of functions
$f(\cdot)$. For example,  if there exists some $\sigma\in\R$ such that
\begin{eqnarray}
\{P\frac{\partial f}{\partial x}(x)\}^{s}\ge\sigma P\label{Pf}
\end{eqnarray}
for all $x\in\R^{n}$ and some $\sigma\in\R$, then
\begin{eqnarray*}
&&(u-v)^{\top}P[f(u)-f(v)]=\int_{0}^{1}(u-v)^{\top}\\
&&\{P\frac{\partial f}{\partial x}(\lambda(v-u)+v)\}^{s}(u-v)d\lambda\ge\sigma (u-v)^{\top}P(u-v),
\end{eqnarray*}
which implies (\ref{vi}). In particular, if the Jacobin of $f(\cdot)$ is
bounded, namely, $ \|{\partial f}/{\partial x}\|\ge J_{f}$ for some $J_f>0$,
then we have
\begin{eqnarray*}
\{P\frac{\partial f}{\partial x}(y)\}^{s}\ge-J_{f}\|P\|,
\end{eqnarray*}
which implies $\sigma=-J_{f}\frac{\|P\|_{2}}{\lambda_{\min}(P)}$ in
(\ref{Pf}).

Second, we consider the differential equations (\ref{r2}) and suppose that the
solutions of (\ref{r2}) satisfy the following inequality:
\begin{eqnarray}
\|(u(t)-u_{0})-(v(t)-v_{0})\|\le\rho(t,\theta,\vartheta,u_{0},v_{0}),\label{r1}
\end{eqnarray}
where $\rho$ is nonnegative-valued continuous map that depends on the node
dynamics map $f(\cdot)$, the initial value $u_{0},v_{0}$ and inputs
$\theta,\vartheta$, and satisfies  $\rho(0,\cdot,\cdot,\cdot,\cdot)\equiv
0$. Geometrically, $\rho$ is an upper-bound estimation of the difference
between the displacements of two trajectories of (\ref{r2}) with respect to
their initial locations,
\begin{align*}
&\|(u(t)-u_{0})-(v(t)-v_{0})\|\\
=&\left\|\int_{0}^{t}[f(u(s))-f(v(s))]d s+(\theta-\vartheta) t\right\|.
\end{align*}

For example, if $f(\cdot)$ is Lipschitz (on the two trajectories):
$\|f(u(s))-f(v(s))\|\le L_{f}\|u(s)-v(s)\|$ for all $s\ge 0$, then we have
\begin{align*}
&\|(u(t)-u_{0})-(v(t)-v_{0})\|\\
\le& L_{f}\int_{0}^{t}\|(u(s)-u_{0})-(v(s)-v_{0})\|ds\\
&+(\|\theta-\vartheta\|+L_f\|u_{0}-v_{0}\|)~t.
\end{align*}
By the Gronwall inequality (\ref{gw2}), we have
\begin{align}
\label{Lips}
\nonumber&\|(u(t)-u_{0})-(v(t)-v_{0})\|\\
\le&\frac{(\|\theta-\vartheta\|+L_f\|u_{0}-v_{0}\|)}{L_{f}}[\exp(L_{f}t)-1].
\end{align}
It can be seen that the upper-bound equals to zero if $t=0$.

It can be seen that the estimation of $\rho$ and $\varrho$ substantially
depends on the form of $f(\cdot)$. There might not be a unified approach to
give precise estimation for general $f(\cdot)$ but might be done case by
case. Therefore, an efficient but cost way is to use integrators that
simulate the node dynamics of $\dot{u}=f(u)+\theta$ as the generators of
$\rho$ and $\varrho$. These generators are independent of the states of the
nodes and so parallel to the networked systems. Figs. \ref{rho}  and
\ref{varrho} show the configurations of the generators of $\rho$ and
$\varrho$ respectively.

\begin{figure}
\centering
\includegraphics[width=0.5\textwidth]{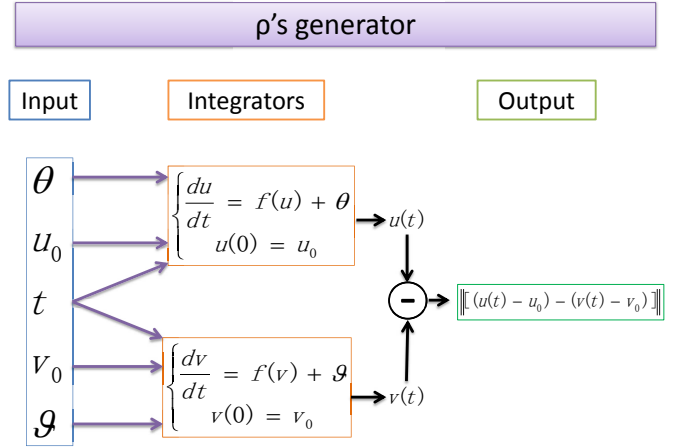}
\caption{$\rho$ generator.} \label{rho}
\end{figure}

\begin{figure}
\centering
\includegraphics[width=0.5\textwidth]{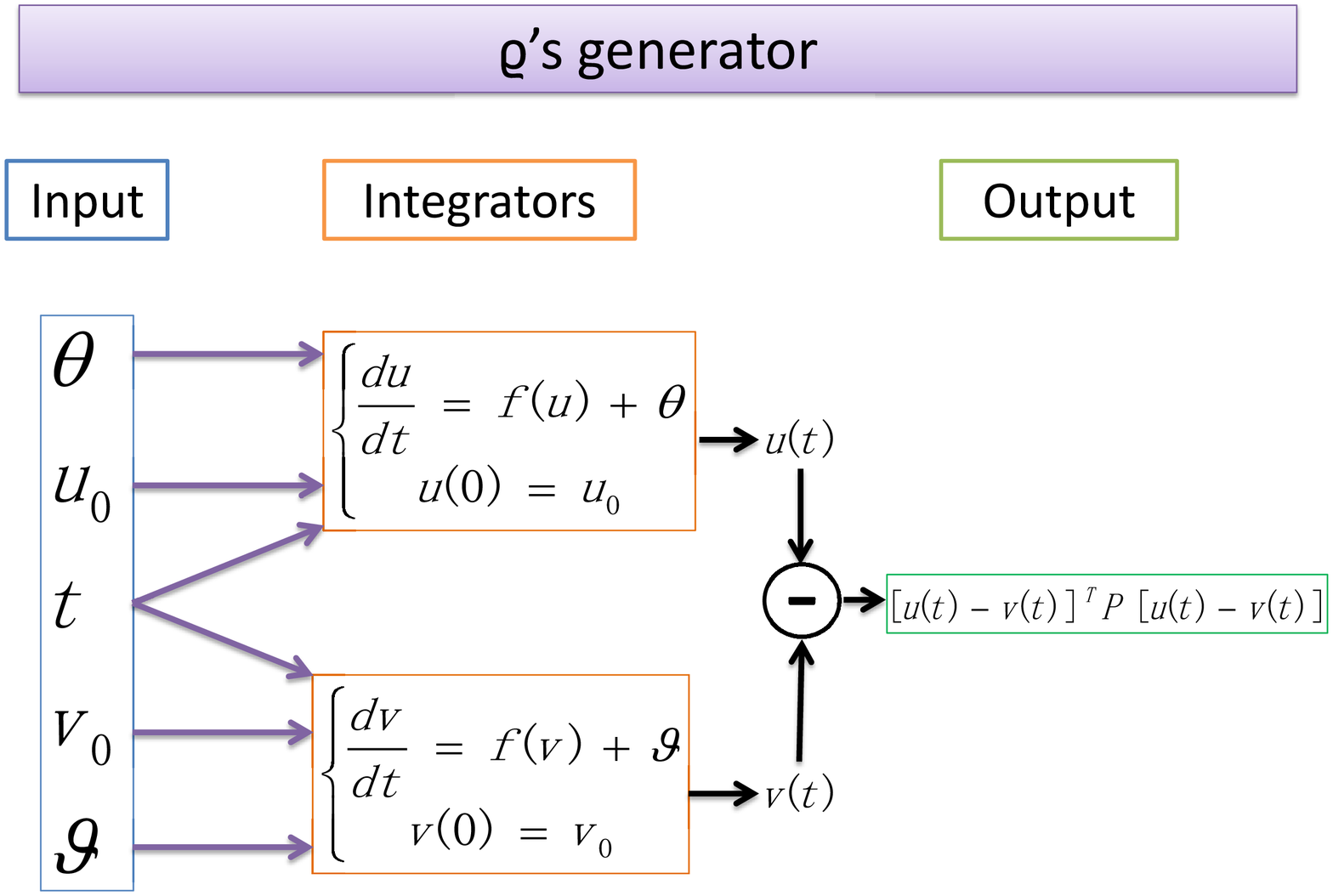}
\caption{$\varrho$ generator.} \label{varrho}
\end{figure}

Let
\begin{align*}
&\theta^{i}_{k}=-c\sum_{j=1}^{m}L_{ij}(x^{j}(t^{i}_{k})-x^{i}(t^{i}_{k})),\\
&\theta_{k_{j}(t)}^{j}=-c\sum_{l=1}^{m}L_{jl}[x^{l}(t^{j}_{k_{j}(t)})-x^{j}(t^{j}_{k_{j}(t)})]
\end{align*}
Then, for $\tau\ge t^{i}_{k}$, we have
\begin{align*}
&\|e_{i}(\tau)\|\le\sum_{q\ne i}(-L_{iq})\|\Gamma\|\rho(\tau-t_{k}^{i},\theta^{i}_{k},\theta_{k_{q}(\tau)}^{q},x^{i}(t_{k}^{i}),x^{q}(t^{i}_{k})),\\
&z_{i}(\tau)\ge\sqrt{\sum_{j\ne i}(-L_{ij})\varrho^{2}(\tau-t_{k}^{i},\theta^{i}_{k},\theta_{k_{j}(\tau)}^{j},x^{i}(t_{k}^{i}),x^{j}(t^{i}_{k}))}.
\end{align*}

With these assumptions and Theorem \ref{thm1}, we have the following result.
\begin{theorem}\label{thm3}
Suppose that $f\in Quad(P,\alpha\Gamma,\beta)$ with positive matrix $P$ and $\beta>0$ such that $P\Gamma$ is semi-positive definite. $c\lambda_{2}(L)>\alpha$. For any positive $\beta'<\beta$, set $\xi^{i}_{k}$ by:
\begin{align}
\xi^{i}_{k}=\max\bigg\{&\xi:~\sum_{q\ne i}(-L_{iq})\|\Gamma\|\rho(\xi,\theta^{i}_{k},\theta_{k_{q}(\xi+t_k^i)}^{q},x^{i}(t_{k}^{i}),x^{q}(t^{i}_{k}))
\nonumber\\
\nonumber&\le\frac{\beta'}{c\sqrt{\lambda_{m}(L)\lambda_m(P)}}\bigg[\sum_{q\ne i}(-L_{iq})\varrho^{2}(\xi,\theta^{i}_{k},\theta_{k_{j}(\xi+t_k^i)}^{j},\\
&~~~~~~~~~~~~~~~~~~
~~~~~~~~~x^{i}(t_{k}^{i}),x^{j}(t^{i}_{k}))\bigg]\bigg\}.\label{event3}
\end{align}
The event timing $\{t^{i}_{k}\}$ are set by the following scheme:
\begin{enumerate}
\item Initialization: $t_{0}^{i}=0$ for all $i=\onetom$;
\item For node $i$, set $\xi^{i}_{k}$ via its neighbor's and its own current states and diffusion by (\ref{event3});
\item If one of its neighbors, for example, $j$, triggers at $t=t^{j}_{k'+1}$ ($k'$ is the latest event at node $j$ before $t$), then  replace $\theta^{j}_{k'}$ by $\theta^{j}_{k'+1}$ in (\ref{event3}) and go to Step 2;
\item Let $t^{i}_{k+1}=t^{i}_{k}+\xi^{i}_{k}$, the event triggers at node $i$ by changing $t^{i}_{k}$ in (\ref{cds}) to $t^{i}_{k+1}$.
\end{enumerate}
Then, system (\ref{cds}) synchronizes.
\end{theorem}
The proof of this theorem can be derived by using (\ref{event1}) in Theorem
\ref{thm1} directly.

\begin{remark}
Firstly,  if node $i$ has a neighbor $j$ satisfying $x^{j}(t^{i}_{k})\ne
x^{i}(t^{i}_{k})$, then $\xi_{k}^{i}$ in (\ref{event3}) is well defined. In
fact, in (\ref{event3}), if $\xi=0$, the left-hand side of the inequality in
(\ref{event3}) equals zero while the right-hand side is nonzero. Therefore, by
the continuous dependence of the parameters in the system (\ref{cds}),
(\ref{event3}) has indeed a positive maximum $\xi$.

Secondly, each node $i$ needs to know the states of itself and
its neighbors. In details, when one node $j$ is triggered, it sends off its new
coupling terms, $\theta_{j}(t^{j}_{k_{j}(t)})$, to all its neighbors for
their updating the estimation of $\xi^{i}_{k}$ for their next updating
times.
\end{remark}

\begin{remark}
In case that $x^{j}(t^{i}_{k})=x^{i}(t^{i}_{k})$ for all neighbors $j$ of node $i$, both left and right sides equal zero, which might lead to a Zeno behavior. To avoid the Zeno behavior, we provide a triggering event, which depends only on $t_{k}^{i}$ by the rule (\ref{event2}) in Theorem \ref{thm1}. Here, we only need to estimate the bounds of $(x^{q}(t)-x^{q}(t_{k}^{l}))-(x^{i}(t)-x^{i}(t_{k}^{l}))$ for any $q,i$ with $L_{iq}\neq 0$. Note
\begin{align*}
&(x^{q}(t)-x^{q}(t_{k}^{l}))-(x^{i}(t)-x^{i}(t_{k}^{l}))\\
=&\int_{t_{k}^l}^t\left[f(x^q(s))-f(x^i(s))+\theta_{k_{q}(s)}^{q}-\theta^{i}_{k}\right]ds.
\end{align*}
Combing with  $\|\theta^{i}_{k}-\theta_{k_{q}(s)}^{q}\| \leq M$, where $M>0$ is some constant, we suppose that
the solutions of (\ref{r2}) satisfy the following inequality:
\begin{eqnarray}
\|(u(t)-u_{0})-(v(t)-v_{0})\|\le\rho_1(t,u_{0},v_{0}).\label{rho2}
\end{eqnarray}
Then, for $\tau\ge t^{i}_{k}$, we have
\begin{align*}
\|e_{i}(\tau)\|\le\sum_{j\ne i}(-L_{ij})\|\Gamma\|\rho_1(\tau-t_{k}^{i},x^{i}(t_{k}^{i}),x^{j}(t^{i}_{k})).
\end{align*}
\end{remark}

\begin{theorem}\label{thm4}
Suppose that $f\in Quad(P,\alpha\Gamma,\beta)$ with positive matrix $P$ and $\beta>0$ such that $P\Gamma$ is semi-positive definite. $c\lambda_{2}(L)>\alpha$. For any positive $\beta'<\beta$, set inter-event interval $\xi^{i}_{k}$ by:
\begin{align}\label{event4}
\nonumber\xi^{i}_{k}=\max\bigg\{&\xi:~\sum_{j\ne i}(-L_{ij})\|\Gamma\|\rho_1(\xi,x^{i}(t_{k}^{i}),x^{j}(t^{i}_{k}))\\
&\le a\exp{[-b(\xi+t_k^i)]}\bigg\}.
\end{align}
The event times $\{t^{i}_{k}\}$ are set by the following scheme:
\begin{enumerate}
\item Initialization: $t_{0}^{i}=0$ for all $i=\onetom$;
\item For node $i$, search $\xi^{i}_{k}$ via its neighbor's and its own current states by (\ref{event3});
\item Triggers node $i$ by changing $t^{i}_{k}$ in (\ref{cds}) to $t^{i}_{k+1}=t^{i}_{k}+\xi^{i}_{k}$,.
\end{enumerate}
Then, system (\ref{cds}) synchronizes.
\end{theorem}

This theorem can be derived by using (\ref{event2}) in Theorem \ref{thm1}.

It should be highlighted that under rule (\ref{event4}), every node does not need to know the coupling terms of neighbors anymore and the inter-event intervals have a lower-bound.

In fact, in (\ref{event4}), if $\xi=0$, the left-hand part equals zero while the right-hand is nonzero. Therefore, according to the continuous dependence of the parameters in the system (\ref{cds}), (\ref{event4}) has indeed a positive maximum argument $\xi_k^i$.

Similar to Theorem \ref{thm2}, we have
\begin{theorem}\label{thm5}
Suppose that $f\in Quad(P,\alpha\Gamma,\beta)$ with positive matrix $P$ and $\beta>0$, satisfies Lipschitz condition with Lipschitz constant $L_f$, and there exists some $\sigma$ (possibly negative) such that (\ref{vi}) holds for all $u,v\in\mathbb R^{n}$ . $c\lambda_{2}(L)>\alpha$ and $P\Gamma$ is semi-positive definite. For any $\beta'<\beta$, any initial condition and any time $t\leq 0$, we have
\begin{enumerate}
\item[(1)]
under the updating rule (\ref{event3}), there exists $\tau_{O}>0$ such that there exists at least one agent
$k\in\{1,\cdots,m\}$ such that the next inter-event interval is strictly
positive and has the lower-bound $\tau_{O}$; in addition, if there exists $\varsigma>0$ such that $z_{i}^{2}(t)\ge \varsigma V(t)$ for all $i=\onetom$ and $t\ge 0$, then the next inter-event interval of every
agent is strictly positive and has a common positive lower-bound.
\item[(2)] suppose $f(\cdot)$ is Lipschitz with constant $L_f$. Then,
under the updating rule (\ref{event4}), the next inter-event interval of every agent is strictly positive and has a common lower-bound $\tau_D'$.
\end{enumerate}

\end{theorem}

\begin{remark}
In comparison with the continuous-time monitoring, the discrete-time
monitoring works well particularly when the states of nodes cannot be
monitored spontaneously. Generally speaking, the main difference between
these two monitoring strategies is that continuous-time monitoring
determines the next updating time in an on-line way, based on the
spontaneous information of states of nodes. Instead, the discrete-time
monitoring predicts the next updating time. Therefore, the discrete-time
monitoring costs less for collecting state information than continuous-time
monitoring. However, as a trade-off, it needs more calculations in
predicting the next updating time, as mentioned in (\ref{event3}) or
(\ref{event4}).
\end{remark}

\section{Examples}
In this section, we present two examples to illustrate the theoretical results. The system is an array of $10$ linearly coupled Chua circuits with the node dynamics
\begin{eqnarray}
f(z)=\left[\begin{array}{c}p*(-z_{1}+z_{2}-g(z_{1}))\\
z_{1}-z_{2}+z_{3}\\
-q*z_{2}\end{array}\right],
\end{eqnarray}
where $g(z_{1})=m_{1}*z_{1}+1/2*(m_{0}-m_{1})*(|z_{1}+1|-|z_{1}-1|)$, with the parameters $p=9.78$, $q=14.97$, $m_{0}=-1.31$ and $m_{1}=-0.75$, which implies that the intrinsic node dynamics (without diffusion) have a double-scrolling chaotic attractor \cite{chua}. The coupling graph topology is shown in Fig. \ref{topology}. $L$ is picked as the Laplacian of the graph where each link has uniform weight $1$. Then, the largest and smallest nonzero eigenvalues equal to $\lambda_{2}=0.8363$ and $\lambda_{m}=7.3484$ respectively. Let $P=\Gamma=I_{3}$. To estimate the parameter $\beta$ in the $Quad$ condition, noting the Jacobin matrices of $f$ is one of the following
\begin{align*}
&A_{1}=\left[\begin{array}{lll}-2.445&9.78&0\\
1&-1&1\\0&-14.97&0\end{array}\right]\\
&A_{2}=\left[\begin{array}{lll}3.0318&9.78&0\\
1&-1&1\\0&-14.97&0\end{array}\right]
\end{align*}
then we can estimate $\beta'=\alpha-\lambda_{\max}((A_{2})^{s})=\alpha-\chi$, where $
\chi=9.1207$ is the upper-bound of the largest eigenvalues of the symmetry parts of all possible Jacobin matrices of $f$.

\begin{figure}
\centering
\includegraphics[width=0.3\textwidth]{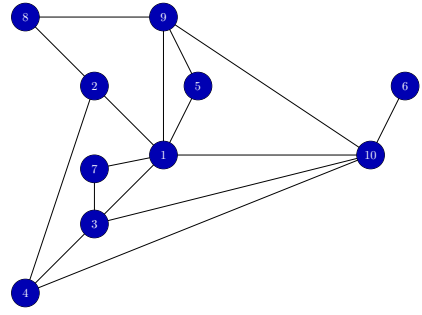}
\caption{Topology of the graph of the coupled system and the pinned node.} \label{topology}
\end{figure}

The ordinary differential equation (\ref{cds}) is numerically solved by the
Euler method with a time step $0.001$ (sec) and the time duration of the
numerical simulations is $[0,2]$(sec).

\subsection{Continuous-time monitoring}

First, we consider the updating rule (\ref{event1}). According to $c\lambda_{2}>\alpha$, where $\lambda_{2}(L)$ is the smallest eigenvalue of $L$, except the unique zero eigenvalue, the supremum of the term $\beta'/(c\sqrt{\lambda_m})$ is estimated as follows:
\begin{eqnarray*}
\sup\frac{\beta'}{c\sqrt{\lambda_m(L)}}=\frac{c\lambda_2(L)-\lambda_{\max}(A_{2}^{s})}{c\sqrt{\lambda_m(L)}}\to\frac{\lambda_{2}(L)}{\sqrt{\lambda_m(L)}}
\end{eqnarray*}
as $~c\to\infty$, by picking $\alpha=c\lambda_{2}(L)$.
Fig. \ref{ets} shows the variation of $\beta'/(c\sqrt{\lambda_m(L)})$ with respect to $c$. In this example, we pick $c=20.3281$, which implies $\beta'/(c\sqrt{\lambda_m(L)})=0.1450$. We employ the updating rule (\ref{event1}) in Theorem \ref{thm1}. Fig. \ref{variation_x_con1} presents the dynamics of each component of the $10$ nodes and show that the coupled system (\ref{cds}) reaches synchronization. Fig. \ref{variation_v_all} shows that $V(t)$ decreases with respect to time and converges toward zero as time goes to infinity.

\begin{figure}
\centering
\includegraphics[width=0.5\textwidth]{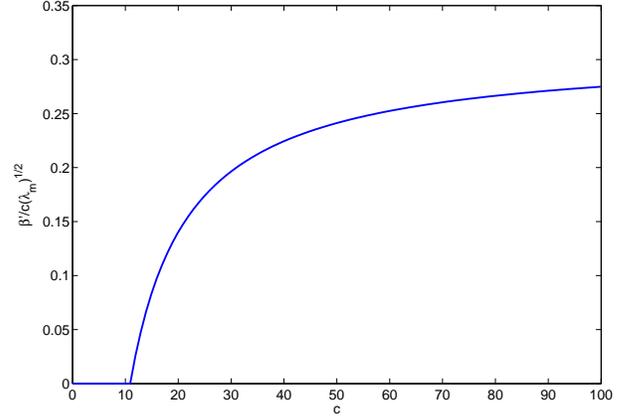}
\caption{Variation of $\beta'/(c\sqrt{\lambda_m})$ under continuous-time monitoring for synchronization with respect to the coupling strength $c$.} \label{ets}
\end{figure}

\begin{figure}
\centering
\includegraphics[width=0.5\textwidth]{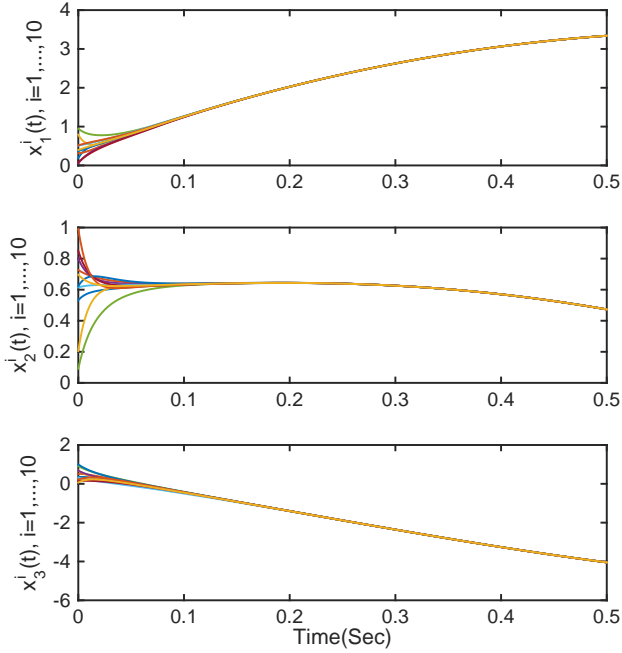}
\caption{Dynamics of components of the coupled system (\ref{cds}) under the event triggering rule (\ref{event1}). }
\label{variation_x_con1}
\end{figure}

Second, we employ the updating rule (\ref{event2}). We take the same value
of $c$ as above and $a=0.5$, $b=0.5$. Fig.
\ref{variation_x_con2} presents the dynamics of each
component of the $10$ nodes and show that the coupled system (\ref{cds})
reaches synchronization. Fig. \ref{variation_v_all} shows that $V(t)$  decreases
with time and converges to zero as time goes to infinity.

\begin{figure}
\centering
\includegraphics[width=0.5\textwidth]{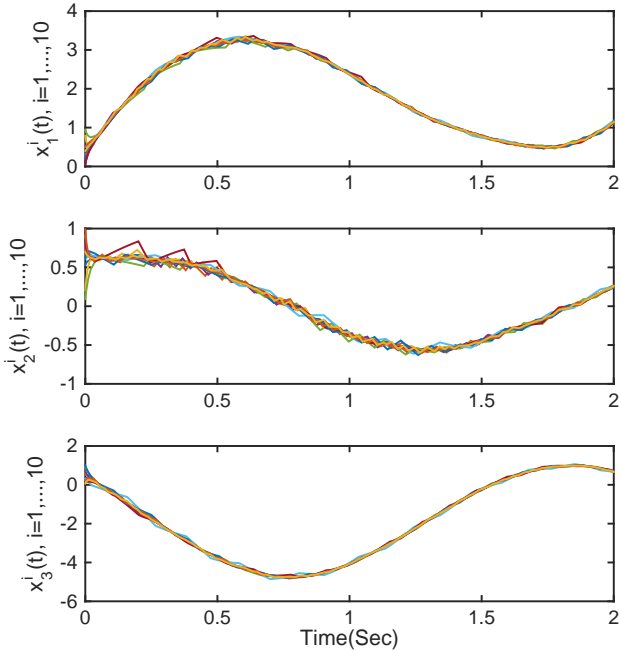}
\caption{Dynamics of components of the coupled system (\ref{cds}) under the event triggering rule (\ref{event2}). }
\label{variation_x_con2}
\end{figure}

\subsection{Discrete-time monitoring}

First, we employ the updating rule (\ref{event3}). The term $\beta'/(c\sqrt{\lambda_{m}})$ can be directly derived from the arguments above. We pick the same $c$ and then $\beta'/(c\sqrt{\lambda_{m}})$ is the same as above. We employ the event trigger algorithm (\ref{event3}) in Theorem \ref{thm3}. Fig. \ref{variation_x_dis1} presents the dynamics of each component of the $10$ nodes and shows that the coupled system (\ref{cds}) reaches synchronisation. Fig. \ref{variation_v_all} shows that $V(t)$ decreases with time and converges to zero.

\begin{figure}
\centering
\includegraphics[width=0.5\textwidth]{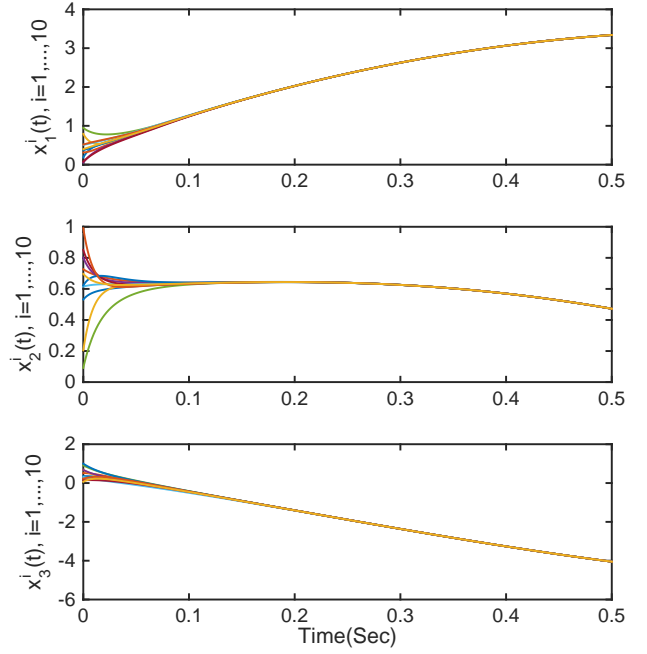}
\caption{Dynamics of components of the coupled system (\ref{cds}) under the event triggering rule (\ref{event3}). }
\label{variation_x_dis1}
\end{figure}

Second, we consider the updating rule (\ref{event4}).  We pick the same $c$ as above and $a=0.5$, $b=0.5$. Fig. \ref{variation_x_dis2} presents the dynamics of each component of the $10$ nodes and show that the coupled system (\ref{cds}) reach synchronization.  Fig. \ref{variation_v_all} shows that $V(t)$ decreases with respect to time and converges toward zero as time goes to infinity.

\subsection{Performance comparison}
In comparison, we consider the original linear coupled system as follows:
\begin{eqnarray}
\frac{dx^{i}(t)}{dt}=f(x^{i}(t))-c\sum_{i=1}^{m}L_{ij}\Gamma(x^{j}(t)-x^{i}(t)),\label{org}
\end{eqnarray}
for $i=\onetom$. By the same setups of model and numerical approach as
above, its performance in terms of converge rates of $V(t)$ is shown
similar with those of event-triggered rules (\ref{event1}) and
(\ref{event3}), as comparatively shown by Fig. \ref{variation_v_all}. As for the performance of rules
(\ref{event2}) and (\ref{event4}), since the exponential convergence rates
are pre-designed, as shown by (\ref{event2}) and (\ref{event4}), it is not
surprising that their convergence rates are not as good as (\ref{org}).
However, their updating times of these rules are much less than those of
rules  (\ref{event1}) and (\ref{event3}), as comparatively shown in Figs. \ref{triggering_times_1}-\ref{triggering_times_2}.

\begin{figure}
\centering
\includegraphics[width=0.5\textwidth]{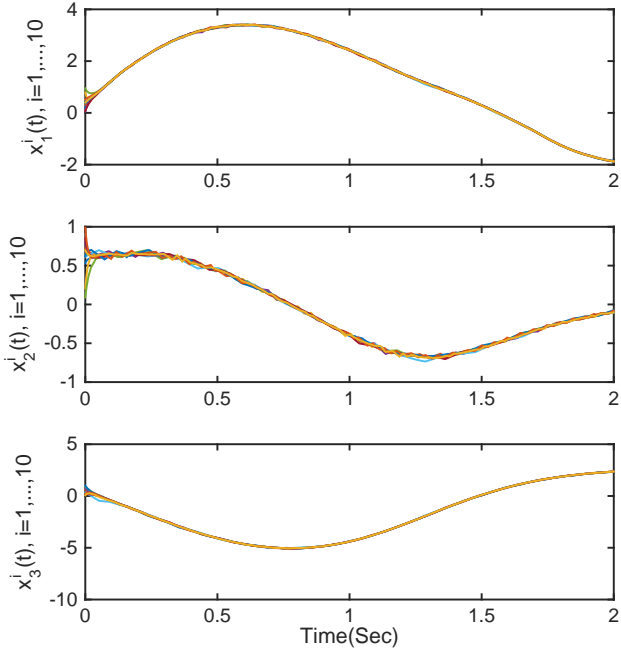}
\caption{Dynamics of components of the coupled system (\ref{cds}) under the event triggering rule (\ref{event4}). }
\label{variation_x_dis2}
\end{figure}

\begin{figure}
\centering
\includegraphics[width=0.5\textwidth]{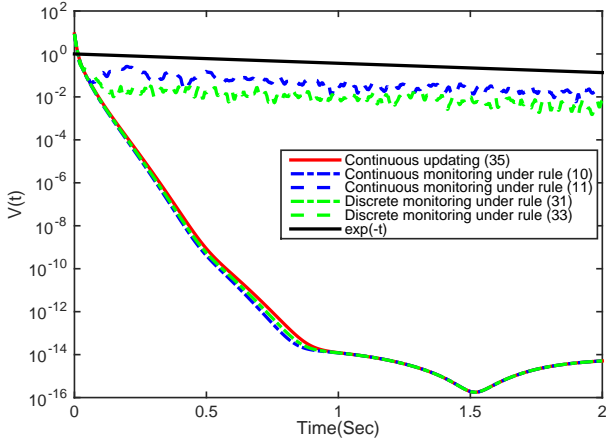}
\caption{Dynamics of $V(t)$ for systems with continuous updating, continuous monitoring under rules (\ref{event1}),(\ref{event2}), discrete monitoring under rules (\ref{event3}), (\ref{event4}).}
 \label{variation_v_all}
\end{figure}

\begin{figure}
\centering
\subfigure[under updating rules (\ref{event1}) and (\ref{event3})]{
\includegraphics[width=.5\textwidth]{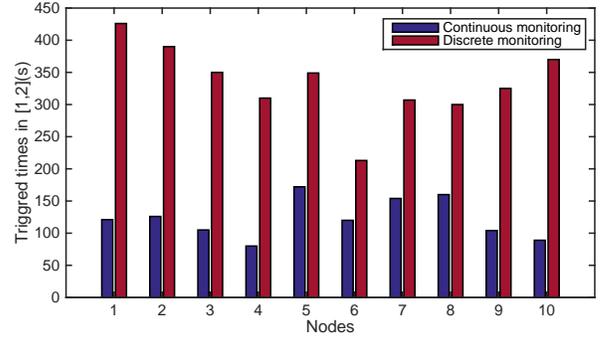}
\label{triggering_times_1}
}
\subfigure[under updating rules (\ref{event2}) and (\ref{event4})]{
\includegraphics[width=.5\textwidth]{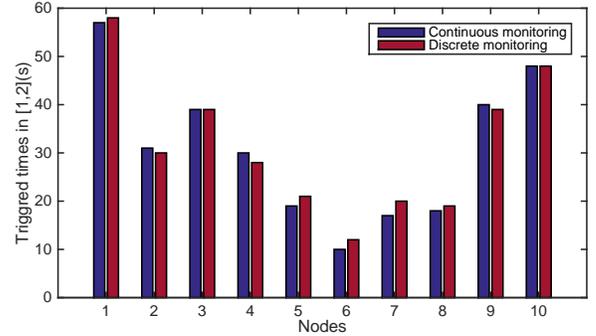}
\label{triggering_times_2}
}
\caption{Histogram of triggering times of each node in $[1,2]$(sec).}
\end{figure}

\section{conclusion}
In this paper, we employed event-triggered coupling configurations to
realize synchronization for linearly coupled dynamical systems. We studied
both continuous monitoring and discrete monitoring schemes: continuous
monitoring scheme means that each node collects its neighborhood's
instantaneous state, and discrete monitoring scheme means that each node
obtains its neighborhood's states at the event triggered time. The
event-triggered rules were proved to perform well and can exclude Zeno
behaviors, as proved for some cases and illustrated by simulations. We showed that there are
trade-offs between better performance in terms of fast convergence and less
updating time slots, and between more cost in observation of states and
more computation load of predicting next updating times. One step further,
there are a few issues, including how to estimate the number of updating
time slots and its dependence on the parameters in the rule and the
structure of network structure, which merits the future research.

\section*{Acknowledgement}
The authors are very grateful to reviewers for their useful comments
and suggestions.

\end{document}